\documentclass[letterpaper, 10 pt, conference]{ieeeconf}  

\IEEEoverridecommandlockouts                            
\title{\LARGE \bf
Constructing Feedback Linearizable Discretizations for Continuous-Time Systems using Retraction Maps.
}

\author{Ashutosh Jindal$^{1}$, Ravi Banavar$^{2}$ and David Mart{\'\i}n Diego$^{3}$
\thanks{*D. Mart{\'\i}n de Diego acknowledges financial support from the Spanish Ministry of Science and Innovation, under grants PID2019-106715GB-C21 and CEX2019-000904-S.}
\thanks{$^{1,2}$
        Systems and Control Engineering, Indian Institute of Technology Bombay, Mumbai, 400076, Maharashtra, India  $^1${\tt\small jindal.ashutosh21@gmail.com}, $^2${\tt\small banavar@iitb.ac.in} }
\thanks{$^{3}$ Instituto de Ciencias Matem\'aticas (CSIC-UAM-UC3M-UCM), Calle Nicol\'as Cabrera 13-15, 28049 Madrid, Spain
{\tt\small david.martin@icmat.es}}
}
\usepackage{mymacros}

\begin{document}

\maketitle
\thispagestyle{empty}
\pagestyle{empty}

\begin{abstract}
Control laws for continuous-time dynamical systems are most often implemented via digital controllers using a
sample-and-hold technique. Numerical discretization of the continuous system is an integral part of subsequent 
analysis. Feedback linearizability of such sampled systems is dependent upon the choice of discretization map
or technique. In this article, for feedback linearizable continuous-time systems, we utilize the idea of retraction maps to construct discretizations that are feedback linearizable as well. We also propose a method to functionally compose discretizations to obtain higher-order integrators that are feedback linearizable.  

\end{abstract}

\section{Introduction}
Digital controllers facilitate the implementation of continuous-time control systems via discretization. For non-autonomous systems i.e., for systems with inputs this is done via (a) sample and hold technique where the control input is held constant between two sampling intervals and (b) a discretization scheme that solves the evolution of the continuous-time dynamical systems numerically. Different numerical schemes result in different discretizations of the continuous time systems. On Euclidean spaces i.e., for systems evolving on $\R[n]$, some of the common numerical integration schemes are Euler Integrations methods, Runge-kutta-based methods, Simpson 1/3 rule, etc. \cite{blanes2017concise}. While these schemes perform well for systems evolving in euclidean spaces, when implemented for systems evolving on general manifolds, they do not guarantee that the system states stay on the manifold. In order to maintain the non-euclidean structure of the underlying manifold one would like to construct integrators that respect the underlying geometry of the continuous-time dynamical system. Such integrators are called geometric integrators and these result in more accurate long-term behavior. A summary of geometric integrator schemes is given in \cite{hairer, blanes2017concise}. Retraction maps  are a generalization of euclidean discretizations on non-euclidean manifolds (see \cite{21MBLDMdD, AbMaSeBookRetraction}). Retraction maps allow us to construct geometric discretizations that guarantee the system states stay on the manifold. 

Feedback linearization allows us to transform a nonlinear control system into a linear system via a coordinate transformation and invertible control feedback. This allows us to utilize the control design methods such as pole placement etc  available for linear systems to synthesize controls for the nonlinear system. A study of feedback linearization for continuous time systems is provided in \cite{brockett1978feedback,jacubczyk1980linearization,isidori1982feedback} and references therein. A discrete-time equivalent of feedback linearization is studied in \cite{grizzle1986feedback,lee1987linearization,jayaraman1993feedback,aranda1996linearization}. Sampling, in general, does not preserve the feedback linearization, i.e., given a feedback linearizable continuous time system, under sample and hold method and a particular choice of discretization the resulting discrete-time system need not be feedback linearizable (in discrete time) in general. In \cite{grizzle1988feedback}, with help of an example shows that a continuous system, under one discretization, may result in a feedback linearizable discrete-time system while it may not do so under some other discretization. Under exact discretization methods, the feedback linearizability of the sampled time systems has been studied in \cite{jakubczyk1987effect,arapostathis1989effect}. Since feedback linearization allows us to utilize the advantages of  linear control theory, it is of interest to find (numerical) discretizations that are feedback linearizable. 

\noindent \textbf{Contribution}: In this article, given a feedback linearization we utilize retraction maps to construct discretizations  that are feedback linearizable. We also provide a way to compose these discretizations to generate symmetric discretizations that are accurate upto second order while maintaining feedback linearizability. However, this requires multi-rate sampling. 

\noindent\textbf{Organization}: The article is organized as follows: in Section \ref{ret_map} and \ref{cts_fed-lin}, we provide a brief introduction on retraction maps and continuous time feedback linearization respectively. The retraction maps are defined for autonomous systems, we extend this notion to nonautonomous systems. In section \ref{num_desc_cts} we provide our main result where we utilize the retractions maps to construct first-order discretizations  that are feedback linearizable and in Section \ref{high_desc} we present a way to compose first-order discretizations to construct higher-order feedback linearizable schemes. This is obtained via multi-rate sampling schemes. We demonstrate these results on a simple example in Section \ref{exmpl}.  
\section{Retraction and  Discretization Maps}
\label{ret_map}
Let $M$ be an $n$ dimensional manifold and $TM$ be the associated tangent bundle. Let $TM\ni(x,v)\mapsto \tau_M(x,v)\coloneqq x$ be the canonical projection onto the manifold. 
\begin{defn}[Discretization Map \emph{(see \cite{21MBLDMdD,AbMaSeBookRetraction})}]
Let $U\subset TM$ be an open neighborhood of the zero section of the tangent bundle $TM$. $U\ni(x,v)\mapsto R(x,y)\in M\times M$ is a discretization map if it satisfies
\begin{enumerate}
    \item $(x,0_x)\mapsto R(x,0_x) = (x,x)$
    \item $T_{(x,0_x)}R^2-T_{(x,0_x)}R^1={\rm Id}_{T_xM} \colon T_{(x,0_x)}T_xM\simeq T_xM \lra T_xM$ is equal to the identity map on $T_xM$ for any $x$ in $M$.
\end{enumerate}
\end{defn}
As a consequence is easy to show that  any discretization map $R$ is a local diffeomorphism. 
\begin{defn}[Adjoint Discretization]
\label{adj_disc}
 Let $R$ be a discretization on $M$. Consider the inversion map $(x,y)\ni M\times M \mapsto \I_M(x,y)\coloneqq (y,x)\in M\times M$. The adjoint of $R$ is defined by $U\ni(x,v)\mapsto R^*(x,v) \coloneqq \I_M(R(x,-v))$. 
\end{defn}
A discretization is called \emph{symmetric} if $R= R^*$. 

Given $X\in\mathfrak{X}$ a vector field on $M$ and a discretization map $R$ we have the following discretization map and a fixed time discretization map $t\mapsto(t-\alpha h, t+(1-\alpha)h)$, $\alpha\in[0,1]$. 
\begin{prop}
\label{ret_desc}
The discretization of $X$ defined by
\[
R^{-1}(x_k, x_{k+1})=hX\underbrace{( \tau_M( R^{-1}(x_k, x_{k+1})) )}_{\in M}
\]
is  a first-order discretization of $X$ and second-order if $R$ is symmetric.
\end{prop}

\begin{prop}
\label{ret_desc_2}
Let $M$ and $N$ be $n$ dimensional manifold and $M\ni x\mapsto \phi(x)\eqqcolon y\in N$ be a diffeomorphism. For a given discretization $R$ on $M$, $R_\phi := (\phi\times\phi)\circ R\circ T\phi^{-1}$ is a  discretization on $N$ (see Figure \ref{fig:my_label}).
\end{prop}
\begin{proof}
For any given $y\in N$ we have that
\begin{eqnarray*}
R_{\phi}(y,0_y)&=&\left( (\phi\times \phi)\circ  R\circ T\phi^{-1}\right)(y,0_y)\\ 
&=&\left((\phi\times \phi)\circ  R\circ T\phi^{-1}\right)(\phi(x),0_{\phi(x)}) 
\\
&=&(\phi\times \phi)^{-1}R(x, 0_x)\\
&=&(\phi\times \phi)^{-1}(x,x) = (y,y)
\end{eqnarray*}
this proves the first condition. 

Now, given a vector $u_y\in T_yN$ we 
\begin{eqnarray*}
& &\left(T_{(y,0_y)}R^2_{\phi}-T_{(y,0_y)}R^1_{\phi}\right)(y,u_y)\\
&=&\odv{}{s}\Big|_{s=0}[
	(\phi\circ R^1\circ T\phi^{-1})(y,su_y)\\&~-&	(\phi\circ R^2\circ T\phi^{-1})(y,su_y)
	    ]\\
	     &=&T_{y}\phi\left(\mat{\odv{}{s}\Big|_{s=0}[
	    R^1(t (T\phi^{-1})(y,u_y)) \\~-	 R^2(t(T\phi^{-1})(y,u_y))
	    ]}\right)\\
	   &=&T_{y} \phi((T\phi^{-1})(y,u_y) ) =(y,u_y)
	   \end{eqnarray*}
using the linearity of the map $T_y\phi$ and that $R$ is a discretization map.
\end{proof}
Using the inversion map $\I_M$ one can easily show that $R^*_\phi = (\phi\times\phi)\circ R^*\circ T\phi^{-1}$ is the adjoint discretization of $R_\phi$. Further, $R_\phi$ is symmetric if $R$ is symmetric. From definition of $R$ and $R_\phi$, Figure \ref{fig:my_label} commutes.
\begin{figure}
    \centering
    \begin{tikzpicture}
  \matrix (m) [matrix of math nodes,row sep=3em,column sep=4em,minimum width=2em]
  {
     TM & TN \\
     M\times M & N\times N \\};
  \path[-stealth]
    (m-1-1) edge node [left] {$R$} (m-2-1)
            edge node [above] {$T\phi$} (m-1-2)
    (m-2-1.east|-m-2-2) edge node [below] {$\phi\times\phi$}
            (m-2-2)
    (m-1-2) edge node [right] {$R_{\phi}$} (m-2-2);
\end{tikzpicture}
    \caption{$R$ and $R_\phi$ commutes as shown above}
    \label{fig:my_label}
\end{figure}
\section{Continuous time control System }
\label{cts_fed-lin}
Let $M$ be an $m$ dimensional manifold and $U\subset\R[n]$ be open. For each $u\in U$ let $X(\cdot,u)\in \mathfrak{X}(M)$ be a vector field on $M$. Then for a fixed $T>0$, a continuous-time dynamical system on $M$ is given by 
\begin{equation}
\label{ctssys}
\odv{}{t}x(t) = X(x(t),u(t)) ~ \forall~ t\in[0,T],
\end{equation}
with $t\mapsto x(t)\in M$ and $t\mapsto u(t)$ for all $t\in[0,T]$. A point $(x_0,u_0)\in M\times U$ is said to be an equilibrium point of \eqref{ctssys} if $X(x_0,u_0) = 0$. 
\subsection{Feedback Linearization of Continuous Time systems}
Let $M$ and $N$ be two $n$-dimensional manifolds and $\phi:M\lra N$ be a diffeomorphism. Let $X\in \mathfrak{X}(M)$ be a vector field on $M$. Then $X_\phi \coloneqq T\phi\circ X\circ\phi^{-1}$ is a vector field on $N$. Further for the dynamical system 
\begin{equation}
\label{ctssys1}
\odv{}{t}y(t) = X_\phi(y(t),u(t))~\forall t\in[0,T]
\end{equation}
with $y(0) = \phi(x(0))$ satisfy $y(t) = \phi(x(t))$  where $x(t)$ is a solution of (\ref{ctssys}) with 
 $t\in[0,T]$.
\begin{defn}[Feedback Linearization]
\label{feedback_linearization}
Let $\Open{x_0}\ni x_0$ and $\Open{u_0}\ni u_0$ be open neighorhoods around $x_0$ and $u_0$ of $M$ and $U$, respectively. Let $\Open{x_0}\ni x\mapsto\phi(x)\coloneqq y\in N:= \R[n]$ be a diffeomorphism to its image and  $\Open{x_0}\times\Open{u_0}\ni (x,u)\mapsto\psi(x,u)\coloneqq v\in\R[m]$ such that for each fixed $x$, $\psi(x,\cdot):U\lra\R[n]$ is invertible.  A given continuous time system \eqref{ctssys} is said to be (locally)  feedback linearizable around $(x_0,u_0)$ on $\Open{x_0}\times\Open {y_0}$ if there exists matrices $A\in\R[n\times n]$ and $B\in\R[n\times m]$ such that $X_\phi (y,u) = Ay+Bv$
with $v = \psi(\phi^{-1}y,u)$. The feedback linearized dynamical system is given by 
\begin{equation}
 \label{ctslin}
 \odv{}{t}y(t) = Ay(t)+Bv(t) ~\forall t\in [0,T].
\end{equation}
For background on feedback linearization, we refer the reader to \cite{brockett1978feedback,jacubczyk1980linearization,isidori1982feedback} and references therein.
\end{defn}
\section{Numerical Discretization of Continuous-Time Systems} 
\label{num_desc_cts}
Continuous time control systems are implemented via digital controllers using \emph{sample and hold} method where the control input $u$ is held constant at a fixed value between two successive samples i.e., $u(t) = u_k~\forall~t\in[t_k,t_k+h[~\forall k\in\N$ where $h$ is the fixed sampling period. Further, since analytical solutions for \eqref{ctssys} are often not available in closed form, the solutions are to be approximated numerically. 
\begin{defn}
Let $U\subset\R[m]$ be open and for each $u\in U$, $X(\cdot,u)\in\mathfrak{X}(M)$ is a vector field on $x$. Let $R$ be a discretization map on $M$ then using Proposition \ref{ret_desc} a discretization of $X(\cdot,u)$ is defined by 
\begin{equation}
\label{disc_sys_1}
    R^{-1}(x_k,x_{k+1}) = hX(\tau_M(R^{-1}(x_k,x_{k+1})),u_k)
\end{equation}
where the control input $u$ is held constant over the interval $[t_k,t_{k+1}[$ i.e., $u(t) = u_k~\forall~ t\in[t_k,T_{k+1}]$.
\end{defn}
A different choice for $R$ leads to different numerical discretization schemes. For example, on Euclidean spaces ($M=\R[m]$),  $R(x,v) = (x,x+v)$ results in the Explicit Euler discretization scheme  $x_{k+1}= x_k+hX(x_k,u_k)$, while $R^*(x,v) = (x-v,x)$ defines the Implicit Euler Discretization with $x_{k+1}-hX(x_{k+1},u_k) = x_k$. Solving \eqref{disc_sys_1} implicitly for $x_{k+1}$, the sampled discrete-time system can be explicitly written as 
\begin{equation}
 \label{disc_sys} 
 x_{k+1} = F(x_k,u_k;h)
\end{equation}
where $x_k\in M$ and $u_k\in U$ for all $k\in\N$ and $M\times U\ni(x,u)\mapsto f(x,u)\in M$ is a smooth map (if $F$ is not well defined on entire $M$ one may very well work with a local definition of $F$, by replacing $M$ with an open neighborhood around $x_0$ in $M$). From the properties of retraction maps, one can show that at equilibrium point $(x_0,u_0)$ one has $F(x_0,u_0;h)= x_0$.  
\subsection{Feedback Linearization of Discrete time system} 
The idea of feedback linearization can be extended to discrete-time systems as follows. Consider the discrete-time system given by \eqref{disc_sys}. 
\begin{defn}[Feedback Linearization (discrete-time)]
Let $\Open{x_0}\ni x_0$ and $\Open{u_0}\ni u_0$ be open neighborhoods around $x_0$ and $u_0$. Let $\Open{x_0}\ni x\mapsto y \eqqcolon \phi(x) \in N\coloneqq \R[n]$ be a diffeomorphism onto its image, and $\Open{x_0}\times\Open{u_0}\ni(x,u)\mapsto v\eqqcolon\psi(x,u)\in\R[m]$ be such that for each $x$, $\psi(x,\cdot)$ is locally invertible.  The discrete-time system \eqref{disc_sys} is said to be feedback linearizable if there exist matrices $A_h$ and $B_h$ such that  
\begin{equation*}
  \phi(F_h(x,u)) = A_h\phi(x) +B_h\psi(x,u) = A_hy +B_hv   
\end{equation*}
The discrete-time system \eqref{disc_sys} is linearized to 
\begin{equation}
\label{disc_lin}
y_{k+1} = A_hy_k+B_hv_k
\end{equation}
\end{defn}
The feedback linearizability of discrete-time systems has been dealt with in great detail in \cite{grizzle1986feedback,lee1987linearization,jayaraman1993feedback,aranda1996linearization}. For sampled time continuous time system the feedback linearizability is in general not preserved, i.e., a feedback linearizable continuous time system when implemented with sample and hold may not result in a feedback linearizable discrete-time system. The linearizability is not only dependent upon the underlying continuous-time system but also on the choice of discretization (see \cite{grizzle1988feedback}). Using this as our motivation we are interested in the following problem -- \textbf{given a feedback linearizable continuous time system \eqref{ctssys} is it possible to construct a numerical discretization \eqref{disc_sys}
that is feedback linearizable?}
\subsection{Constructing feedback linearizable discretization maps}
Consider the continuous time system given by \eqref{ctssys}. Let $\phi$ and $\psi$ be as in Definition \eqref{feedback_linearization}. Suppose \eqref{ctssys} is feedback linearizable to \eqref{ctslin}. Keeping $v(t) = v_k$ for all $t\in[t_k,t_{k+1}[$, let $R$ be a discretization for \eqref{ctslin} such that it preserves the linearity of \eqref{ctslin} i.e. it results in a discrete system
\begin{equation*}
    y_{k+1} = A_hy_k+B_hv_k
\end{equation*}
 where $A_h$ and $B_h$ are fixed matrices of appropriate order. Given a discretization map $R$ for \eqref{ctslin}, using Proposition \ref{ret_desc_2} one can construct a discretization map 
 \begin{equation}
 \label{4.7}
     R_{\phi^{-1}} = (\phi\times\phi)^{-1}\circ R\circ T\phi
 \end{equation}
and a discretization scheme
\begin{equation}
\label{disc_sys_2}
R^{-1}_{\phi^{-1}}(x_k,x_{k+1}) = hX(\tau_M(R_{\phi^{-1}}^{-1}(x_k,x_{k+1}),u_k))
\end{equation}
\begin{thm}
\label{lin_res}
Let $R$ be a discretization map  for the feedback linearized continuous system \eqref{ctslin} preserving linearity, then the nonlinear system \eqref{ctssys} has a discretization given by \eqref{disc_sys_2} that is feedback linearizable in the discrete-time domain. 
\end{thm}
\begin{proof}
Define $y_k = \phi(x_k)$ and $\psi(x_k,u_k) = v_k$ for all $k$. From \eqref{4.7} we have
\begin{eqnarray*}
    R_{\phi^{-1}}^{-1}(x_k,x_{k+1}) &=& (T\phi^{-1}\circ R^{-1}\circ(\phi\times\phi))(x_k,x_{k+1})\\
    &=& (T\phi^{-1}\circ R^{-1})(y_k,y_{k+1}),
\end{eqnarray*}
and 
\begin{eqnarray*}
&&X(\tau_M(R_{\phi^{-1}}^{-1}(x_k,x_{k+1})),u_k)\\
&=&X\big(\tau_M((T\phi^{-1}\circ R^{-1})(y_k,y_{k+1})),u_k)\big)\\
&=&X_\phi\big(\tau_N(R^{-1}(y_k,y_{k+1})), u_k)\big)
\end{eqnarray*}
From \eqref{disc_sys_2} we have $(x_k,x_{k+1}) = R_{\phi^{-1}}(hX(\tau_M(R_{\phi^{-1}}^{-1}(x_k,x_{k+1})))$ and therefore we have
\begin{eqnarray*}
&&(\phi\times\phi)(x_k,x_{k+1})\\&=& (\phi\times\phi)\circ R_{\phi^{-1}}(hX(\tau_M(R_{\phi^{-1}}^{-1}(x_k,x_{k+1})),u_k)\\
&=&R\big((hX_\phi(\tau_N(R^{-1}(y_k,y_{k+1})),u_k)\big)
\end{eqnarray*}
and therefore we have
\begin{eqnarray*}
(y_k,y_{k+1}) = R(hX_\phi(\tau_N(R^{-1}(y_k,y_{k+1})),u_k)).
\end{eqnarray*}
But since $R$ preserves linearity and  $X_\phi (y,u) = Ay+Bv$
with $v = \psi(\phi^{-1}y,u)$, we have 
\begin{eqnarray*}
    y_{k+1} &=& A_hy_k +B_h\psi(x_k,u_k)\\
    &=& A_hy_k+B_hv_k
\end{eqnarray*}
This concludes the fact that discretization given by \eqref{disc_sys_2} is feedback linearizable under the coordinate change $x\mapsto\phi(x) = y$ and a modified control input $(x,u)\mapsto \psi(x,u) = v$. 
\end{proof}
\begin{rmk}
It is important to note that independent of the order of $R$ one can  ensure an accuracy of $R_{\phi}$ up to the first order. This is due to the fact that while implementing \eqref{ctssys} via the sample and hold, the control input $u$ is to be held constant on the interval $[t_k,t_{k+1}[$. This is in general not possible while simultaneously keeping the linearized control input $v$ constant over $[t_k,t_{k+1}[$ as $v(t) =\psi(x(t),u(t))$. Instead of employing the exact control input $u(t)$ over the interval, we apply the control $u_k$ satisfying $v_k = \psi(x_k,u_k)$ for all $t$ where $x_k$ is the state sampled at $t = t_k$.
\end{rmk}
\subsection{Linearizability of Adjoint discretization}
Given a discretization map, $R_{\phi^{-1}}$ one can construct an adjoint discretization $R^*_{\phi^{-1}}$ as given by the Definition \ref{adj_disc}. From proposition \ref{ret_desc_2} we have 
\begin{equation}
\label{disc_len_232}
    R^*_{\phi^{-1}} = (\phi\times\phi)^{-1}\circ R^*\circ T\phi. 
\end{equation}
Using the definition of $R^*$ and the inversion map, $R^*$ induces a following discretization scheme 
\begin{equation*}
 (R^*)^{-1}(y_{k+1},y_k) = -hX_\phi(\tau_N((R^*)^{-1}(y_{k+1},y_k),v  _k)).
\end{equation*}
Suppose $R$ is such that it preserves the linearity of \eqref{ctssys} for $-h$ as well i.e.,  the following discretization
\begin{eqnarray*}
 R^{-1}(y_k,y_{k+1}) = -hX_\phi(\tau_N(R(y_k,y_{k+1})),u_k)  
\end{eqnarray*}
results in a linear discrete-time system of the form
\begin{equation}
    y_{k+1} = A_{-h}y_k + B_{-h}v_k
\end{equation}
then we have the following result.
\begin{thm}
\label{lin_res_adj}
 Let $R$ be a discretization of \eqref{ctslin} preserving linearity for $h$ as well as $-h$. Let $R^*$ be the adjoint of $R$, then $R^*_{\phi^{-1}}$ given by \eqref{disc_len_232} results in a discretization 
\begin{equation}
     x_{k+1} = F^*(x_{k},u_{k};h)
\end{equation}
i.e., feedback is linearizable. Moreover, the linearizing coordinate is given by $x\mapsto\phi(x):=y$ and the linearized system is given by 
\begin{equation}
\label{disc_sys_lin2}
    y_k = A_{-h}y_{k+1}+B_{-h}v_k
\end{equation}
with $v_k = \psi(x_{k+1},u_k)$. 
\end{thm}
The proof of the above theorem follows a similar process to that of Theorem \ref{lin_res} and is hence omitted. Moreover, the control input $u_k
$ is to be calculated implicitly from the control input $v_k$. Similar to $R_{\phi^{-1}}$, $R_{\phi^{-1}}^*$ is also accurate upto first order. 
\section{Constructing Higher Order Discretizations}
\label{high_desc}
\begin{defn}[Global and truncated error] Consider the continuous time system \eqref{ctssys}, then for a given discretization \eqref{disc_sys} the $n$-step global error is defined as $$e(k) \coloneqq x(t_k)-x_k,$$ where $x(t_k)$ is the exact solution of \eqref{ctssys} evaluated at $t= hk$. The one-step truncated error at $t_k$ is 
\begin{equation*}
    \tilde{x}_k = \left(F(x(t_k),u_k))-x(t_k)\right)/{h}
\end{equation*}
\end{defn}
For an $r$-order discretizations, the one-step truncation error is bounded above by $\norm{\tilde{x}_k}\leq K\norm{h}^{r}$ \cite{suli2003introduction}. For first-order methods, the error varies linearly with the stepsize, therefore one requires a smaller stepsize to have better accuracy. For instance, in order to have an accuracy of an order of $10^{-4}$ the stepsize $h$ is to be of the order of $10^{-4}$, whereas for second-order methods a stepsize of an order of $10^{-2}$ shall suffice. The discretizations $R$ in Definition \ref{ret_desc} are in general first order. However, if $R$ is symmetric then it is accurate up to second accurate. This serves as our motivation to construct symmetric discretization.  
\subsection{Symmetric Discretizations}
Let $M$ be an $n$ dimensional manifold and $X\in\mathfrak{X}(M)$ be a vector field on $M$. Let $R$ be a discretization map on $M$ and $R^*$ be its associated adjoint. Define by composition a discretization scheme as follows :
\begin{equation}
\label{symm_int}
\begin{split}
R^{-1}(x_k,x_{k+1/2}) &= \frac{h}{2}X(\tau_M(R^{-1}(x_k,x_{k+1/2})))\\    
(R^*)^{-1}(x_{k+1/2},x_{k+1}) &= \frac{h}{2}X(\tau_M((R^*)^{-1}(x_{k+1/2},x_{k+1})))
\end{split}
\end{equation}
In the above equation, $x_{k+1/2}\in M$ is to be taken as an intermediate point and is to be solved implicitly to get a discrete system of type \eqref{disc_sys}. For any $x,y\in M$ the following is true 
\begin{equation}
\label{5.2}
  (R^*)^{-1}(x,y) =  I_{TM}\circ R^{-1}(y,x)  
\end{equation} and
\begin{equation}
\label{5.3}
  X(\tau_M((R^*)^{-1}(x,y)))=X(\tau_M(R^{-1}(y,x)))  
\end{equation}
where $ TM\ni(x,v_x)\mapsto I_{TM}(x,v_x)=(x,-v_x)\in TM$. 
\begin{prop}
Discretization given by \eqref{symm_int} is symmetric.
\end{prop}
\begin{proof}
Replacing $(x_k,x_{k+1})$ and $h$ with $(x_{k+1},x_k)$ and $-h$ in \eqref{symm_int} we have
\begin{equation*}
\begin{split}
R^{-1}(x_{k+1},x_{k+1/2}) &= -\frac{h}{2}X(\tau_M(R^{-1}(x_{k+1},x_{k+1/2})))\\    
(R^*)^{-1}(x_{k+1/2},x_k) &= -\frac{h}{2}X(\tau_M((R^*)^{-1}(x_{k+1/2},x_{k})))
\end{split}
\end{equation*}
Using \eqref{5.2} and \eqref{5.3} we get
\begin{align*}
(R^*)^{-1}(x_{k+1/2},x_{k+1})&= \frac{h}{2}X(\tau_M((R^*)^{-1}(x_{k+1/2},x_{k+1}))),\\
R^{-1}(x_k,x_{k+1/2})&= \frac{h}{2}X(\tau_M(R^{-1}(x_k,x_{k+1/2})))
\end{align*}
i.e., \eqref{symm_int}, thereby proving \eqref{symm_int} is symmetric.
\end{proof}
\begin{rmk}
Since \eqref{symm_int} is symmetric and therefore is second order.    
\end{rmk}
For nonautonomous systems, the control input $u_k$ is held constant between $t \in [t_k,t_{k+1}[$, \eqref{symm_int} is then modified as
\begin{equation}
 \label{symm_int_na}
 \begin{split}
R^{-1}(x_k,x_{k+1/2}) &= \frac{h}{2}X(\tau_M(R^{-1}(x_k,x_{k+1/2})),u_k)\\    
(R^*)^{-1}(x_{k+1/2},x_{k+1})\\ = &\frac{h}{2}X(\tau_M((R^*)^{-1}(x_{k+1/2},x_{k+1})),u_k)
\end{split}
\end{equation}
Under closed-loop performance i.e., applying a feedback control $u_k =u(x_k)$ \eqref{symm_int} loses its symmetric nature. This can be overcome by employing multirate sampling methods.
\subsection{Multirate Sampling}
\begin{defn}[Multirate Sampling]
Consider a continuous time system given by \eqref{ctssys}. Let $h$ be the sampling time interval i.e., $x_k= x(t_k)$ and $t_{k+1} = t_k+h$. For a fixed $N\in \{1,2,\ldots,n\}$, and for each $i\in\{1,2,\ldots N\}$ let $(x,u)\mapsto F_i(x,u)\eqqcolon F_i^u(x)$ be discretizations of \eqref{ctssys}. The $N^{th}$ step evolution is then given by 
\begin{equation}
  x_{k+N} =F_N^{u_{k+N-1}}\circ\ldots\circ F_2^{u_{k+1}}\circ F_1^{u_k}(x_k)   
\end{equation}
Sampling states $x_k$ at a rate $N$ times that of control input $u_k$ we get a multistep discretization given by 
\begin{equation}
 x_{k+N} = \bar F(x_k,u_k,\ldots u_{k+N-1}) 
\end{equation}
The control input $u_k,\ldots,u_{k+N-1}$ are to be computed a priori at $t_k$ are functions of the state $x_{Nm}$, $Nm<k$.
\end{defn}
Setting $N=2$ and $F_1$ and $F_2$ as $F$ and $F^*$ from \eqref{disc_sys} and \eqref{4.7} respectively. Under multirate sampling, the discrete system generated by \eqref{symm_int_na} is given by 
\begin{eqnarray*}
    x_{k+1/2} &=& F(x_k,u_k;h/2)\\
    x_{k+1/2} &=& F(x_{k+1},u_{k+1};-h/2).
\end{eqnarray*}
Let $u_k = u(x_k)$ be a closed-loop control input for discretization \eqref{disc_sys_1}. Setting $u_{k+1} = u(x_{k+1})$ renders \eqref{symm_int_na} symmetric and the discretization is given by 
\begin{equation}
\label{symm_int_Na1}
 F(x_k,u(x_k);h/2) =F(x_{k+1},u(x_{k+1});-h/2)  
\end{equation}
which is symmetric and therefore is accurate up to second order. Corresponding continuous time control input is 
\begin{eqnarray*}
    u(t) = \begin{cases}
    &u_k,~ \quad t\in[t_k,t_{k+1/2}[\\
    &u_{k+1},~\quad t\in[t_{k+1/2},t_{k+1}[
    \end{cases}
\end{eqnarray*}
where $t_{k+1/2} = t_k+\frac{h}{2}$ and $t_{k+1}= t_k+h$.
\begin{thm}
\label{thm5.2}
Consider the continuous time system given by \eqref{ctssys}. Let $R_\phi$ be its discretization as given by \eqref{4.7} and \eqref{disc_sys_2} be its associated discretization. Then one can construct a symmetric discretization given by \eqref{symm_int_na}, the resulting discrete system given by
\begin{eqnarray}
\label{symm_disc_sys}
F(x_k,u_k;h/2) = F(x_{k+1},u_{k+1};-h/2)  
\end{eqnarray}
is symmetric and is of second order. Moreover, \eqref{symm_disc_sys} is feedback linearizable under coordinates $x\mapsto \phi(x)\eqqcolon y$ and the modified control input is given by $(x,u)\mapsto \psi(x,u)\eqqcolon v$. The linearized system is given by 
\begin{equation}
\label{disc_symm_lin}
    A_{h'}y_{k} + B_{h'}v_k =  A_{-h'}y_{k+1} + B_{-h'}v_{k+1}
\end{equation}
where $h' = h/2$.
\end{thm}
\begin{proof}
\eqref{symm_disc_sys} being symmetric is trivial from the definition itself. Since $\phi$ is a local diffeomorphism we have 
\begin{eqnarray*}
    \phi(F(x_k,u_k;h')) &=& \phi(F(x_{k+1},u_{k+1});-h')\\
A_{h'}\phi(x_k)+B_{h'}\psi(x_k,u_k) &=& A_{-h'}\phi(x_{k+1})+\\ &~&~B_{-h'}\psi(x_{k+1},u_{k+1})\\
    A_{h'}y_{k}+B_{h'}v_k &=&  A_{-h'}y_{k+1} + B_{-h'}v_{k+1},
\end{eqnarray*}
thereby completing the proof.
\end{proof}
\begin{rmk}
The control input $v_k$, $v_{k+1}$ can be computed apriori at $t=t_k$ from $z_k$. The control input $u_k$ and $u_{k+1}$ are than computed implicitly solving $\psi(x_k,u_k)= v_k$, with $x_k =\phi^{-1}(z_k)$. 
\end{rmk}
\begin{rmk}
 Theorem \ref{thm5.2} is different from the result in \cite{grizzle1988feedback} in the sense that here the rate of multi-sampling is fixed apriori while \cite{grizzle1988feedback} the order of discretization is chosen so that the resulting scheme is feedback linearizable.
\end{rmk}
\section{Example}
\label{exmpl}
 In order to demonstrate the ideas discussed we consider the following example. Consider the following dynamical system evolving on $M=\R[2]$ and $U=\R$. 
 \begin{equation}
 \label{ex1}
     \odv{}{t}\pmat{x_1(t)\\x_2(t)} = \pmat{a\sin(x_2(t))\\-(x_1(t))^2+u(t)}
 \end{equation}
 where $a\in\R\setminus\{0\}$ is fixed and given. Define $\phi(x_1,x_2) =(x_1,a\sin(x_2))\coloneqq (y_1,y_2)$ and $\psi(x_1,x_2,u) =(-(x_1)^2+u)a\cos(x_2) \coloneqq v$. The linearized system is then given by 
 \begin{equation}
 \label{ex_lin}
     \odv{}{t}\pmat{y_1(t)\\y_2(t)} = \pmat{y_2(t)\\v(t)}.
 \end{equation}
 Choosing the Explicit Euler discretization i.e., $\R[2]\times\R[2]\ni(y,w)\mapsto R(y,w) = (y,y+w))\in\R[2]\times\R[2]$, the discrete system is given by
 \begin{equation}
 \label{exln_dis}
 \begin{split}
 y_{1,k+1} &= y_{1,k}+hy_{2,k},\\
 y_{2,k+1} &= y_{2,k}+hv_k.
 \end{split}
 \end{equation}
 Lifting $R$ via $\phi^{-1}$ to get a discretization for \eqref{ex1}, $R_{\phi^{-1}}$ induces the following discretization scheme
 \begin{equation}
 \tag{EES}
 \label{ex_dis_el}
 \begin{split}
 x_{1,k+1} &= x_{1,k}+ha\sin(x_{2,k})\\
 x_{2,k+1} &= \arcsin(\sin(x_{2,k})+h(-x_{1,k}^2+u_k)\cos(x_{2,k})).
 \end{split}
 \end{equation}
which can be compactly written as $x_{k+1}= F(x_k,u_k;h)$ with $F((x_1,x_2),u;h) = (x_1+ha\sin(x_2),\arcsin(\sin(x_2)+h(-x_1^2+u)\cos(x_2)))$. It is easy to see that \eqref{ex_dis_el} is feedback linearizable for around the equilibrium point $(0,0,0)$. 

The associated adjoint scheme $R^*(y,w) = (y-w,w)$ defines the Implicit Euler Discretization. Lifting $R^*$, $R_{\phi^{-1}}^*$ induces the following discretization scheme 
\begin{equation}
\tag{IES}
\label{IES1}
x_{k} = F(x_{k+1},u_k;-h)
\end{equation}
A symmetric integrator can be defined for Example \eqref{ex1} by composing $R_{\phi^{-1}}$ and $R_{\phi^{-1}}^*$. Using multi-rate sampling with states sampled at a 
 rate twice that of control. The symmetric integrator for \eqref{ex1} is given by 
 \begin{equation}
\tag{SES}
 \label{ex_dis_sym}
 F(x_{k+1},u_{k+1};-h) = F(x_{k},u_k;h).
\end{equation}
One can check that \eqref{ex_dis_sym} is also feedback linearizable and the linearized system is given by 
\begin{equation}
 A_{-h'}y_{k+1} +B_{-h'}v_{k+1} = A_{h'}y_{k}+B_{h'}v_k  
\end{equation}
with $A_{h'} = \pmat{1&h'\\0&1}$ and $B_{h} = \pmat{0\\h'}$ for $h'= h/2$ and $v_k = \psi(x_k,u_k)$.

The three schemes were simulated under the following parameters: For all three schemes, we had $a=1$, and the initial condition was chosen as $x(0) = (0.25,\pi/6)$. For $K=\bmat{-10&-10}$, the control schemes were chosen as in Table \ref{tab:my-table}. 
\begin{table}
\begin{centering}
\begin{tabular}{|c|c|}
\hline
\textbf{Discretization}& \textbf{Associated Control}\\ \hline\eqref{ex_dis_el} & $v_k = Ky_k$\\\hline
\eqref{IES1} & $v_k = Ky_{k+1}$\\ \hline\eqref{ex_dis_sym} & $v_k = Ky_k,~~ v_{k+1} = Ky_{k+1}$  \\ \hline
\end{tabular}
\caption{\centering{Control input for various discretization schemes. ($K = \bmat{-10&-10}$)}}
\label{tab:my-table}
\end{centering}
\end{table}
The schemes were simulated for various step sizes and the error was compared with the standard ODE solver (ODE45) available in MATLAB. The simulation was run for $t\in[0,5]$ and the trajectories for the discrete-time system \eqref{ex_dis_el} and continuous-time system \eqref{ex1} are plotted in Figure \ref{fig-state}. The corresponding error is plotted in Figure \ref{fig-error}. The control input is plotted in \ref{fig-control}. Similarly, the system trajectory  for \eqref{ex_dis_sym} is plotted in Figure \ref{fig-state_sym}, and the control input and error in Figures\ref{fig-control_sym} and \ref{fig-error_sym} respectively. To show the multi-rate sampling,  A zoomed-in version of the control signal around $t_k=0$ is also plotted in figure \ref{fig-zoom_sym}. It can be seen that the absolute error for \eqref{ex_dis_sym} is significantly smaller than that of \eqref{ex_dis_el}. In Figure \ref{fig-error_pct}, we compare the (percentage) relative error $100*\norm{e(t_k)}/\norm{x(t_k)}$, for \eqref{ex_dis_el}, \eqref{IES1} and \eqref{ex_dis_sym}. While percentage error increases for \eqref{ex_dis_el} and \eqref{IES1}, the absolute error is actually quite small (the increase is partly because of the precision errors when $x$ becomes small). For various stepsizes,  Table \ref{tab:my-table1} records the order of the error magnitude for the various discretization. It can be seen that \eqref{ex_dis_sym} outperform \eqref{ex_dis_el} and \eqref{IES1} and the error is proportional to $h^2$. 
\begin{table}
\begin{centering}
\begin{tabular}{l|lll|}
\cline{2-4}
 & \multicolumn{3}{l|}{Order of error magnitude}              \\ \hline
\multicolumn{1}{|l|}{Stepsize} & \multicolumn{1}{l|}{\eqref{ex_dis_el}} & \multicolumn{1}{l|}{\eqref{IES1}} & \eqref{ex_dis_sym} \\ \hline
\multicolumn{1}{|l|}{$h = 10^{-1} $} & \multicolumn{1}{l|}{$10^{-1}$} & \multicolumn{1}{l|}{$10^{-1}$} & $10^{-2}$  \\ \hline
\multicolumn{1}{|l|}{$h= 10^{-2}$}  & \multicolumn{1}{l|}{$10^{-3}$} & \multicolumn{1}{l|}{$10^{-3}$} & $10^{-4}$  \\ \hline
\multicolumn{1}{|l|}{$h=10^{-3}$}   & \multicolumn{1}{l|}{$10^{-4}$} & \multicolumn{1}{l|}{$10^{-4}$} & $10^{-6}$ \\ \hline
\multicolumn{1}{|l|}{$h=10^{-5}$}   & \multicolumn{1}{l|}{$10^{-5}$} & \multicolumn{1}{l|}{$10^{-4}$} & $10^{-8}$ \\ \hline
\end{tabular}%
\caption{\centering{Order of error magnitude for various step sizes.}}
\label{tab:my-table1}
\end{centering}
\end{table}
\begin{figure}
\begin{center}
\begin{tikzpicture}
\begin{axis}[
    width=\linewidth,
    height=5cm,
    at ={(0,-0.25)},
    font=\sf,
    xlabel={$t_k$ (s)},
    ylabel={ },
    xmin= 0, xmax=5,
    ymin= -0.4, ymax=0.6,
    ytick={-0.4,-0.2,0.0,0.2,0.4,0.6},
    xtick={0,1,2,3,4,5},
    legend pos=north east,
    no markers,
    grid=none,
    minor tick num =5,
    minor tick style={draw=none},
    minor grid style={thin,color=black!10},
    major grid style={thin,color=black!10},
    tick align=outside,
    axis x line*=none,
    axis y line*=none,
]
\addlegendentry{$x_{1,k}$}
\addplot[color=blue,smooth,thick] table [x=t,y=X1,col sep=comma]{data_ees.csv};
\addlegendentry{$x_{2,k}$}
\addplot[color=red,smooth,thick] table [x=t,y=X2,col sep=comma]{data_ees.csv};
\addlegendentry{$x_1(t_k)$}
\addplot[color=green,smooth,dashed,thick] table [x=t,y=x1,col sep=comma]{data_ees.csv};
\addlegendentry{$x_2(t_k)$}
\addplot[color=black,smooth,dashed,thick] table [x=t,y=x2,col sep=comma]{data_ees.csv};
\end{axis}
\end{tikzpicture}
\caption{System State $x_k$ for \eqref{ex_dis_el} plotted against exact discretization (ODE45) $x(t_k)$ for stepsize $h=10^{-2}$ and $t_k\in[0,5]$.} 
\label{fig-state}
\end{center}
\end{figure}
\begin{figure}
\begin{center}
\begin{tikzpicture}
\begin{axis}[
    width=\linewidth,
    height=5cm,
    at ={(0,-0.25)},
    font=\sf,
    xlabel={$t_k$ (s)},
    ylabel={ },
    xmin= 0, xmax=5,
    ymin= -10, ymax=2.0,
    ytick={-8,-6,-4,-2,0,0.0},
    xtick={0,1,2,3,4,5},
    legend pos=south east,
    no markers,
    grid=none,
    minor tick num =5,
    minor tick style={draw=none},
    minor grid style={thin,color=black!10},
    major grid style={thin,color=black!10},
    tick align=outside,
    axis x line*=none,
    axis y line*=none,
]
\addlegendentry{$u_{k}$}
\addplot[color=blue,smooth,thick] table [x=t,y=U1,col sep=comma]{data_ees.csv};
\addlegendentry{$u(t_k)$}
\addplot[color=red,smooth,dashed,thick] table [x=t,y=u1,col sep=comma]{data_ees.csv};
\end{axis}
\end{tikzpicture}
\caption{Control input $u_k$ for \eqref{ex_dis_el} plotted against exact discretization (ODE45) $u(t_k)$ for stepsize $h=10^{-2}$ and $t_k\in[0,5]$.} 
\label{fig-control}
\end{center}
\end{figure}
\begin{figure}
\begin{center}
\begin{tikzpicture}
\begin{axis}[
    width=\linewidth,
    height=5cm,
    at ={(0,-0.25)},
    font=\sf,
    xlabel={$t_k$ (s)},
    ylabel={},
    xmin= 0, xmax=5,
    ymin= 0, ymax=0.006,
    ytick={0,0.001,0.002,0.003,0.004,0.005},
    xtick={0,1,2,3,4,5},
    legend pos=south east,
    no markers,
    grid=none,
    minor tick num =5,
    minor tick style={draw=none},
    minor grid style={thin,color=black!10},
    major grid style={thin,color=black!10},
    tick align=outside,
    axis x line*=none,
    axis y line*=none,
]
\addlegendentry{$\norm{e(k)}$ }
\addplot[color=blue!80,smooth,thick] table [x=t,y=hn,col sep=comma]{data_ees.csv};
\end{axis}
\end{tikzpicture}
\caption{Magnitude of global error for \eqref{ex_dis_el} for a stepsize $h=10^{-2}$ and $t_k\in[0,5]$.} 
\label{fig-error}
\end{center}
\end{figure}
\begin{figure}
\begin{center}
\begin{tikzpicture}
\begin{axis}[
    width=\linewidth,
    height=5cm,
    at ={(0,-0.25)},
    font=\sf,
    xlabel={$t_k$ (s)},
    ylabel={ },
    xmin= 0, xmax=5,
    ymin= -0.4, ymax=0.6,
    ytick={-0.4,-0.2,0.0,0.2,0.4,0.6},
    xtick={0,1,2,3,4,5},
    legend pos=north east,
    no markers,
    grid=none,
    minor tick num =5,
    minor tick style={draw=none},
    minor grid style={thin,color=black!10},
    major grid style={thin,color=black!10},
    tick align=outside,
    axis x line*=none,
    axis y line*=none,
]
\addlegendentry{$x_{1,k}$}
\addplot[color=blue,smooth,thick] table [x=t,y=X1,col sep=comma]{data_sym.csv};
\addlegendentry{$x_{2,k}$}
\addplot[color=red,smooth,thick] table [x=t,y=X2,col sep=comma]{data_sym.csv};
\addlegendentry{$x_1(t_k)$}
\addplot[color=green,smooth,dashed,thick] table [x=t,y=x1,col sep=comma]{data_sym.csv};
\addlegendentry{$x_2(t_k)$}
\addplot[color=black,smooth,dashed,thick] table [x=t,y=x2,col sep=comma]{data_sym.csv};
\end{axis}
\end{tikzpicture}
\caption{System State $x_k$ and for \eqref{ex_dis_sym} plotted against exact discretization (ODE45) $x(t_k)$ for stepsize $h=10^{-2}$ and $t_k\in[0,5]$.} 
\label{fig-state_sym}
\end{center}
\end{figure}
\begin{figure}
\begin{center}
\begin{tikzpicture}
\begin{axis}[
    width=\linewidth,
    height=5cm,
    at ={(0,-0.25)},
    font=\sf,
    xlabel={$t_k$ (s)},
    ylabel={ },
    xmin= 0, xmax=5,
    ymin= -10, ymax=2.0,
    ytick={-8,-6,-4,-2,0,0.0},
    xtick={0,1,2,3,4,5},
    legend pos=south east,
    no markers,
    grid=none,
    minor tick num =5,
    minor tick style={draw=none},
    minor grid style={thin,color=black!10},
    major grid style={thin,color=black!10},
    tick align=outside,
    axis x line*=none,
    axis y line*=none,
]
\addlegendentry{$u_{k}$}
\addplot[color=blue,smooth,thick] table [x=t,y=U,col sep=comma]{data_sym_cont.csv};
\addlegendentry{$u(t_k)$}
\addplot[color=red,smooth,dashed,thick] table [x=t,y=u,col sep=comma]{data_sym_cont.csv};
\end{axis}
\end{tikzpicture}
\caption{Control input $u_k$ for \eqref{ex_dis_sym} plotted against exact discretization (ODE45) $u(t_k)$ for stepsize $h=10^{-2}$ and $t_k\in[0,5]$.} 
\label{fig-control_sym}
\end{center}
\end{figure}
\begin{figure}
\begin{center}
\begin{tikzpicture}
\begin{axis}[
    width=\linewidth,
    height=5cm,
    at ={(0,-0.25)},
    font=\sf,
    xlabel={$t_k$ (s)},
    ylabel={},
    xmin= 0, xmax=5,
    ymin= 0, ymax=0.0006,
    ytick={0,0.0001,0.0002,0.0003,0.0004,0.0005},
    xtick={0,1,2,3,4,5},
    legend pos=south east,
    no markers,
    grid=none,
    minor tick num =5,
    minor tick style={draw=none},
    minor grid style={thin,color=black!10},
    major grid style={thin,color=black!10},
    tick align=outside,
    axis x line*=none,
    axis y line*=none,
]
\addlegendentry{$\norm{e(k)}$ }
\addplot[color=blue!80,smooth,thick] table [x=t,y=hn,col sep=comma]{data_sym.csv};
\end{axis}
\end{tikzpicture}
\caption{Magnitude of global error for \eqref{ex_dis_sym} for a stepsize $h=10^{-2}$ and $t_k\in[0,5]$.} 
\label{fig-error_sym}
\end{center}
\end{figure}
\begin{figure}
\begin{center}
\begin{tikzpicture}
\begin{axis}[
    width=\linewidth,
    height=5cm,
    at ={(0,-0.25)},
    font=\sf,
    xlabel={$t_k$ (s)},
    ylabel={ },
    xmin= 0, xmax=0.05,
    ymin= -10, ymax=-4,
    ytick={-10,-8,-7,-6,-5},
    xtick={0,0.005,0.01,0.015,0.02,0.025,0.03,0.035,0.04,0.045},
    legend pos=south east,
    no markers,
    grid=none,
    minor tick num =5,
    minor tick style={draw=none},
    minor grid style={thin,color=black!10},
    major grid style={thin,color=black!10},
    tick align=outside,
    axis x line*=none,
    axis y line*=none,
]
\addlegendentry{$u_{k}$}
\addplot[color=blue,const plot,thick] table [x=t1,y=U1,col sep=comma]{data_sym_cont.csv};
\addlegendentry{$u(t_k)$}
\addplot[color=red,const plot,dashed,thick] table [x=t1,y=u1,col sep=comma]{data_sym_cont.csv};
\end{axis}
\end{tikzpicture}
\caption{Zoomed-in Control input $u_k$ for \eqref{ex_dis_el}  stepsize $h=10^{-2}$ and $t_k\in[0,0.05]$.} 
\label{fig-zoom_sym}
\end{center}
\end{figure}
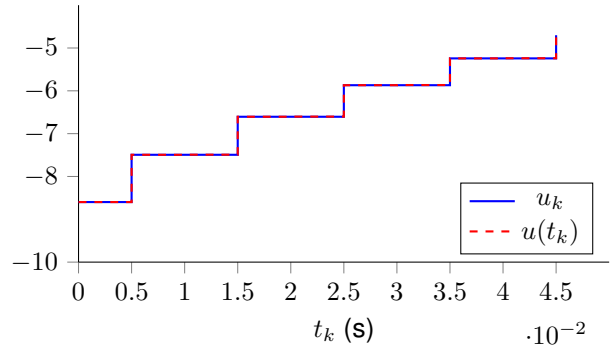
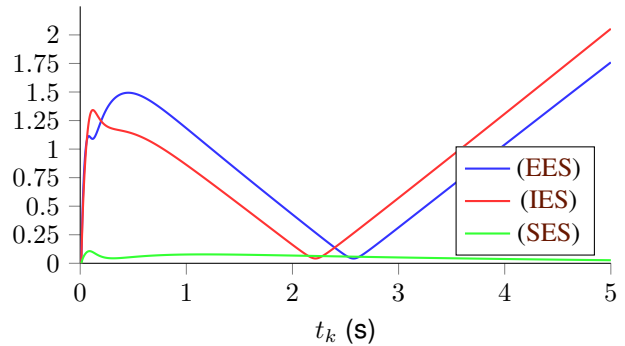
\begin{figure}
\begin{center}
\begin{tikzpicture}
\begin{axis}[
    width=\linewidth,
    height=5cm,
    at ={(0,-0.25)},
    font=\sf,
    xlabel={$t_k$ (s)},
    ylabel={},
    xmin= 0, xmax=5,
    ymin= 0, ymax=2.25,
    ytick={0,0.25,0.5,0.75,1.0,1.25,1.5,1.75,2.0},
    xtick={0,1,2,3,4,5},
    legend pos=south east,
    no markers,
    grid=none,
    minor tick num =5,
    minor tick style={draw=none},
    minor grid style={thin,color=black!10},
    major grid style={thin,color=black!10},
    tick align=outside,
    axis x line*=none,
    axis y line*=none,
]
\addlegendentry{\eqref{ex_dis_el}}
\addplot[color=blue!80,smooth,thick] table [x=t,y=hr,col sep=comma]{data_ees.csv};
\addlegendentry{\eqref{IES1}}
\addplot[color=red!80,smooth,thick] table [x=t,y=hr,col sep=comma]{data_ies.csv};
\addlegendentry{\eqref{ex_dis_sym}}
\addplot[color=green!80,smooth,thick] table [x=t,y=hr,col sep=comma]{data_sym.csv};
\end{axis}
\end{tikzpicture}
\caption{Comparing (percentage) relative error for \eqref{ex_dis_el}, \eqref{IES1} and \eqref{ex_dis_sym} for $h=10^{-2}$} 
\label{fig-error_pct}
\end{center}
\end{figure}
\section{Conclusions}
\label{cnclsn}
In this article, we have utilized the idea of retraction maps and their lifts under diffeomorphism to construct feedback linearizable discretization. Given a continuous-time feedback linearizable system, we show that one can build first-order discretization that preserves feedback linearizability. This is done by lifting a discretization of the linearized continuous time system. We have also shown a way to functionally compose two first-order discretizations to design second-order discretizations that are feedback linearizable. However, this comes at the cost of multi-rate sampling.
\section*{Acknowledgment}
We would like to thank Debasish Chatterjee
(Professor, Systems and Control Engineering, IIT Bombay) for his ideas and  discussions which were crucial in the development of this work.
\bibliographystyle{IEEEtran}
\bibliography{IEEEabrv,ref}

\end{document}